\documentclass[letterpaper, 10 pt, conference]{ieeeconf}  % Comment this line out if you need a4paper

\IEEEoverridecommandlockouts                              % This command is only needed if 
\usepackage[T1]{fontenc}
\usepackage{xcolor}
\usepackage{cite}
\usepackage{amsmath,amssymb,amsfonts}

\usepackage{amsthm}
\usepackage{nameref}
\usepackage{bbm}
\usepackage{algorithmic}
\usepackage{graphicx}
\usepackage{textcomp}
\usepackage{dsfont}
\usepackage{float}
\usepackage{todonotes}

\let\labelindent\relax
\usepackage{enumitem}

\usepackage{tikz} 
\usepackage{tikz-network}
\usetikzlibrary{arrows.meta,calc, angles, quotes}
\usepackage{subcaption}
\usepackage{standalone}
\usepackage{tikz}
\usepackage{tikz-network}
\usetikzlibrary{arrows,intersections,calc,positioning, angles, quotes}
\usetikzlibrary{decorations.markings}
\usetikzlibrary{shapes.misc}
\usetikzlibrary{patterns}
\usepackage{amsmath}
\usepackage{amsfonts}
\usepackage{cleveref}
\usepackage{mathtools}
\usepackage{upgreek}
\usepackage{float}
\usepackage{empheq}
\usepackage[theorems,skins]{tcolorbox}
\newtcolorbox{mymathbox}[1][]{colback=white, sharp corners, #1}
\usepackage[dvipsnames]{xcolor}

\newtheorem{definition}{Definition}

\newtheorem{theorem}{Theorem}

\newtheorem{proposition}{Proposition}

\newtheorem{example}{Example}

\newcommand{\continuation}{??}

\newtheorem{assumption}{Assumption}

\newcommand{\Null}{\mathrm{Null}}
\newcommand{\IM}{\mathrm{IM}}

\newcommand{\Aut}{\mathrm{Aut}}
\global\long\def\rr{\mathcal{R}}

\global\long\def\G{\mathcal{G}}
\global\long\def\E{\mathcal{E}}
\global\long\def\V{\mathcal{V}}

\global\long\def\R{\mathbb{R}}
\global\long\def\C{\mathcal{C}}

\title{\LARGE \bf
{Formation Control via Rotation Symmetry Constraints}}

\author{Zamir Martinez  and Daniel Zelazo, \IEEEmembership{Senior Member, IEEE}% <-this % stops a space
\thanks{This work was supported by the Israel Science Foundation grant no. 453/24 and the Gordon Center for Systems Engineering.}
\thanks{Zamir Martinez and Daniel Zelazo are with the Stephen B. Klein Faculty of Aerospace Engineering, Technion – Israel Institute of Technology, Haifa 3200003, Israel (e-mails: z.m@campus.technion.ac.il, dzelazo@technion.ac.il.)}
}

\begin{document}

\maketitle
\thispagestyle{empty}
\pagestyle{empty}

%%%%%%%%%%%%%%%%%%%%%%%%%%%%%%%%%%%%%%%%%%%%%%%%%%%%%%%%%%%%%%%%%%%%%%%%%%%%%%%%
\begin{abstract}

This work introduces a distributed formation control strategy for multi-agent systems based solely on rotation symmetry constraints. We propose a potential function that enforces inter-agent \textbf{rotational} symmetries, whose gradient defines a control law that drives the agents toward a desired planar symmetric configuration. We show that only $n-1$ edges (the minimal connectivity requirement) are sufficient to implement the strategy, where $n$ is the number of agents. We further augment the design to address the \textbf{maneuvering problem}, enabling the formation to undergo coordinated translations, rotations, and scaling along a predefined virtual trajectory. Simulation examples are provided to validate the effectiveness of the proposed method.
\end{abstract}

%%%%%%%%%%%%%%%%%%%%%%%%%%%%%%%%%%%%%%%%%%%%%%%%%%%%%%%%%%%%%%%%%%%%%%%%%%%%%%%%
\section{INTRODUCTION}

The demand for distributed formation control schemes in multi-agent systems (MAS) has grown significantly in recent years, with applications ranging from UAV swarm coordination for mapping and surveillance \cite{Martinez2006} to satellite constellation coordination for efficient communication relays \cite{Chenyu2024}. The role of a formation control scheme is to provide a control law that steers the agents into a desired spatial configuration in a distributed fashion. This is commonly achieved by imposing explicit geometric constraints between neighboring agents. For example, distance-based schemes \cite{Krick2009,OH2015424} fix inter-agent distances, while bearing-based schemes \cite{ZHAO2019_CSM} fix relative directions. In both cases, the desired target configuration is characterized using \textit{only} local information shared between neighboring agents. 

In many formations, the desired configuration exhibits spatial symmetries between agents such as rotations and/or reflections, often inherent to sensing coverage or communication requirements. The work \cite{Zelazo2025forced} introduced an approach that leverages formation symmetries together with inter-agent distance constraints, drastically reducing the required inter-agent communication links compared with traditional distance-based approaches. This motivates the question of whether it is possible to design a formation control scheme that relies solely on symmetry constraints.

Graph theory provides a natural framework for modeling the decentralization, interaction topology, and geometric configuration of a MAS. Agents are represented as nodes (vertices), with communication links as edges. A central challenge lies in balancing sparse information exchange while ensuring convergence to the desired configuration. To address this challenge, distance-based and bearing-based schemes leverage results from rigidity theory, relying on minimal infinitesimal rigidity (MIR) as a crucial architectural requirement to guarantee convergence to a desired shape \cite{Krick2009, ZHAO2019_CSM}. For distance-based approaches in $\mathbb{R}^2$, the MIR property requires at least $2n-3$ edges to uniquely determine the formation (up to translations, rotation, and so-called flip ambiguities), where $n$ is the number of agents. 

A closely related line of research builds on augmented Laplacian formulations, such as complex-Laplacian approaches \cite{deMarina2020, Zhou2025, Trinh_AUT2018}, where complex weights replace the standard scalar weights of the graph Laplacian and encode inter-agent rotations, translations, and scalings. Matrix-weighted Laplacians provide a related framework, where structured matrix weights encode inter-agent geometric relations. This idea has appeared in several formation control settings: bearing-based control \cite{ZHAO2019_CSM}, where projection matrices enforce relative direction constraints, and more recent works \cite{Fan2025}, where matrix weights play the same role as complex weights. Together, these works highlight the potential of matrix-weighted Laplacians for achieving formation control objectives. 

Motivated by this line of research, we build on \cite{Zelazo2025forced} to study formations defined \textit{solely} through symmetric relations. Specifically, consider a group of $n$ agents whose desired configuration is characterized by prescribed rotational symmetries. These constraints are modeled in Euclidean space via cyclic point-group isometries \cite{alt94}, enforced between designated agent pairs. Assume that each agent has access to its own state and may exchange this information \textit{only} with neighboring agents, as determined by an \textit{undirected} interaction graph. The control objective is to design a distributed control law that drives the agents from any initial state to a desired configuration satisfying the required rotation symmetry constraints.
%
%This paper provides a foundation for solving the formation control problem solely under rotation constraints, providing a complementary approach with rigidity-based methods \cite{Krick2009, Zhao2019}. 
This is achieved by introducing a potential function that enforces rotational symmetries between neighboring agents, whose gradient yields a distributed control law that drives the system toward the null space of a symmetry-constraining matrix-weighted Laplacian. We show that $n-1$ edges (the minimal connectivity requirement) are sufficient to guarantee convergence to the desired formation. To enhance flexibility, we present an augmentation of the control strategy, enabling the desired formation to be achieved while undergoing coordinated translations, rotations, and scalings according to a time-varying reference virtual trajectory, thereby addressing the formation maneuvering problem. Additionally, the effectiveness of the approach is demonstrated through a numerical example extension in $\R^3$.

The paper is organized as follows. Section \ref{sec.symmetry} reviews the mathematics of symmetry, focusing on graphs and frameworks. Section \ref{sec.formation} introduces the symmetry constrained formation control problem and presents the controller. Section \ref{sec.maneuver} extends the controller to allow for formation maneuvers, while Section \ref{sec.R3} demonstrates numerically an extension to $\R^3$. Finally, concluding remarks are offered in Section \ref{sec.conclusion}.

\paragraph*{Notations} 
A graph $\G=(\V,\E)$ consists of two non-empty sets: $\V=\{1,...,n\}$, the set of nodes, and $\E\subseteq \V \times \V$, the set of edges. In this work, $\G$ is assumed to be undirected. The notation $ij\in\E$ indicates that agent $i$ can receive information from neighboring agent $j$, and vice versa. Let $I_n\in\mathbb{R}^{n\times n}$ be the identity matrix, and let $\mathds{1}_n \in \mathbb{R}^n$ be the all-one column vector of dimension $n$. Let $\otimes$ denote the Kronecker product.

\section{Symmetry in Graphs and Frameworks}\label{sec.symmetry}

The main focus of this work is to leverage the inherent symmetries of a desired configuration to solve the formation control problem. In this direction, we first review notions from group theory and graph theory used to formally define symmetry.
\subsection{Symmetry in Graphs}
Group theory provides a powerful mathematical framework for describing symmetry. In the context of graphs, symmetries correspond to structure-preserving transformations of the vertex set—formally captured by the notion of \emph{automorphisms}. The collection of all such transformations forms a group, known as the \emph{automorphism group} of the graph.
We begin by briefly recalling the definition of a group.
\begin{definition}
A \emph{group} is a set $\Gamma$ equipped with a binary operation $\circ$ such that:
\begin{itemize}
    \item \textbf{Closure:} For all $a,b \in \Gamma$, the composition $a \circ b$ is also in $\Gamma$.
    \item \textbf{Associativity:} $(a \circ b) \circ c = a \circ (b \circ c)$ for all $a,b,c \in \Gamma$.
    \item \textbf{Identity:} There exists an element $\mathrm{id} \in \Gamma$ such that $a \circ \mathrm{id} = \mathrm{id} \circ a = a$ for all $a \in \Gamma$.
    \item \textbf{Inverses:} For each $a \in \Gamma$, there exists an inverse $a^{-1} \in \Gamma$ such that $a \circ a^{-1} = a^{-1} \circ a = \mathrm{id}$.
\end{itemize}
The \emph{order} of a group is the number of its elements. A subset $B\subseteq\Gamma$ that is itself a group under $\circ$ is called a \emph{subgroup}.
\end{definition}

In the setting of graphs, these ideas appear naturally when considering automorphisms.
\begin{definition}
Let $\G = (\V, \E)$ be a finite, simple graph. An \emph{automorphism} of $\G$ is a permutation $\psi : \V \to \V$ such that
$$
    uv \in \E \quad \Leftrightarrow \quad \psi(u)\psi(v) \in \E.
$$
\end{definition}
That is, an automorphism preserves the adjacency structure of the graph. The identity permutation $\mathrm{id}$ is always an automorphism, and if $\psi$ is an automorphism, then so is its inverse $\psi^{-1}$. Moreover, the composition of two automorphisms is again an automorphism. These properties ensure that the set of all automorphisms of $\G$ forms a group under composition. This group is called the \emph{automorphism group} of $\G$, denoted by $\Aut(\G)$.
One can express every permutation as a composition of disjoint cycles. A cycle is a successive action of the permutation that sends a vertex back to itself, i.e., $i\to \psi(i) \to \psi(\psi(i)) \to \cdots \to \psi^k(i)=i$, where $\psi^k = \underbrace{\psi \circ \cdots \circ \psi}_{k \text{ times }}$. Such a cycle is compactly written using the \emph{cycle notation}, denoted by $(i\,\psi(i)\,\cdots \psi^{k-1}(i))$. The integer $k$ is the \emph{length} of the cycle.
\begin{definition}\label{def:Gamma_sym_def}
A graph $\G$ is \emph{$\Gamma$-symmetric} for any subgroup $\Gamma \subseteq \Aut(\G)$.
\end{definition}
\begin{example} \label{ex:aut3}
Fig.~\ref{fig:c3_ex1} shows the cycle graph $C_3$. We can identify all the automorphisms of $\Aut(C_3)$. We first identify the identity permutation $\mathrm{id}$. Additionally, consider a counter-clockwise rotation by $120^\circ$ of $C_3$. This gives the automorphism (in cycle notation) $\psi_1=(1\,2\,3)$. We also have $\psi_2=\psi_1^2=\psi_1\circ \psi_1=(1\,3\,2)$, which can be interpreted geometrically as a rotation by $240^\circ$.
Additional permutations can be found by considering reflections. Consider first the reflection about the vertical blue line, giving the permutation $\psi_3=(1)(2 \, 3)$. Similarly, the reflection about the red line yields $\psi_4=(3)(1\,2)$, and the reflection about the green line gives $\psi_5=(2)(1\, 3)$. Thus, $\Aut(C_3)=\{\mathrm{id},\psi_1,\ldots,\psi_5\}$ has 6 automorphisms.
\begin{figure} [!h]
\vspace{-0.4cm}
    \begin{center}    \includegraphics[width=0.22\linewidth]{figures/c3_ex1.tex}
    \end{center}
\caption{Cycle graph $C_3$, with $6$ automorphisms in $\Aut(\G)$.}\label{fig:c3_ex1}
\end{figure}
\end{example}
\vspace{-0.2cm}

Note that the choice $\Gamma=\{\mathrm{id}, \psi_1,\psi_1^2\}$ corresponds to the subgroup of rotational automorphisms of $C_3$. In this case, $C_3$ can be considered as a $\Gamma$-symmetric graph, where any vertex can be mapped to any other under the rotation actions of $\Gamma$.
\subsection{Symmetry in frameworks}\label{sec:symmetry_fwks}
The embedding of symmetric graphs in Euclidean space is of interest, especially for formation control problems. In this direction, we now consider symmetry of frameworks \cite{Bernd2017sym}. A framework in $\mathbb{R}^2$ is defined as the pair $(\G,p)$, where $p: \V\rightarrow\mathbb{R}^2$ assigns each node in $\G$ a position in Euclidean space, which represents the physical position of the agents in the network. 

\begin{definition}\label{def:tau-gamma}
Let $\Gamma$ be represented as a \emph{point group}, i.e., a subgroup of the orthogonal group $O(\mathbb{R}^2)$, via a homomorphism $\tau:\Gamma\rightarrow O(\mathbb{R}^2)$, which assigns to each $\gamma\in\Gamma$ an isometry in $\mathbb{R}^2$. A framework $(\G,p)$ is called \emph{$\tau(\Gamma)$-symmetric} if 
\begin{equation}\label{eq:symfwk}
\tau(\gamma) p_i=p_{\gamma (i)} \quad \forall \gamma\in \Gamma ,\quad i\in \V.
\end{equation}
\end{definition}

In this work, we restrict our study to frameworks whose underlying graph $\G$ is the cycle graph $C_n$. Using the standard Schoenflies notation for point groups \cite{alt94,atk70}, we consider the rotational symmetries described by the cyclic point group $\mathcal{C}_n$ of order $n \geq 1$. That is, $\mathcal{C}_n$ specifies the rotation symmetries that map agents to one another under rotations about the origin.

We define $\Gamma_r\in\Aut(C_n)$ as the subgroup of rotational automorphisms of $C_n$, where each element $\tau(\gamma)$ is represented as a rotation about the origin by an angle $\theta=2\pi/n$. Then, in a planar setting (i.e., in $\mathbb R^2$), $\tau(\Gamma_r)$ coincides with the cyclic point group $\mathcal{C}_n$. We represent the elements $\tau(\gamma)$ by the standard rotation matrix $R(\theta)\in SO(2).$
%$$R(\theta)=\begin{bmatrix}
%    \cos(\theta) & -\sin(\theta) \\ \sin(\theta) & \cos(\theta)
%\end{bmatrix}\in SO(2).$$
Thus, for any two vertices $i,j\in\V$ of a $\mathcal{C}_n$-symmetric framework, we denote by $\gamma_{ij}\in\Gamma_r$ the group element satisfying $\tau(\gamma_{ij})p_i=p_j$ with $\tau(\gamma_{ij})=R(\theta)$ and consequently, $\tau(\gamma_{ji})=R(\theta)^T$. In each case, the desired configuration satisfies condition \eqref{eq:symfwk}, where agent positions are mapped to one another by the corresponding rotation in $\mathcal{C}_n$.
 
\begin{figure}[!b]
\begin{center}
\vspace{-0.2cm}
\includegraphics[width=0.75\linewidth]{figures/cn_fw_ex1.tex}
    \end{center}
\vspace{-0.2cm}
\caption{Symmetric frameworks with $C_n$ as the underlying graph. (a) and (b) are $\mathcal{C}_4$-symmetric, and (c) is $\mathcal{C}_6$-symmetric.}
\label{fig:sym_fw_Cn}
\end{figure}
\section{Symmetry-based Formation Control}\label{sec.formation}

We consider a team of $n$ agents modeled by the integrator dynamics
\begin{align}\label{int-dynamics}
\dot p_i(t) = u_i(t), \quad i \in \{1, \ldots, n\},
\end{align}
where $p_i(t) \in \mathbb{R}^2$ is the position of agent $i$ and $u_i(t) \in \mathbb{R}^2$ is its control to be designed. The coordination objective we consider is for the agents to arrange themselves into a configuration characterized by a specific symmetry class - rotation relationships between neighboring agents. Assume the desired configuration is a $\mathcal{C}_n$-symmetric framework with the cycle graph $C_n$ as the underlying graph. 

We define an \emph{interaction graph} $\G_I = (\V,\E_I)$ to specify which 
agents are able to exchange information. This graph is defined as a spanning tree subgraph of $C_n$, to ensure the minimal connectivity requirement between agents in the MAS. For example, for a $\mathcal{C}_4$-symmetric 
formation in Fig. \ref{fig:sym_fw_Cn}(a), the interaction 
edge set $\E_I = \{12,23,34\}$ satisfies the connectivity
requirement.

Let $\Gamma_r \subseteq \Aut(C_n)$ denote the subgroup of 
rotational automorphisms of $C_n$, and let $\G_I=(\V,\E_I)$ be the interaction graph defined as a spanning tree subgraph of $C_n$. The control objective is to design a distributed control law $u_i(t)$ for each agent $i$ such that, for every edge $ij \in \mathcal{E}_I$, we have
\begin{align}\label{sym_problem}
    \lim_{t \to \infty} \big\| p_i(t) - \tau(\gamma_{ji})\, p_j(t) \big\| = 0.
\end{align}
Here, $\gamma_{ji} \in \Gamma_r$ is the permutation mapping $j$ to $i$ and $\tau(\gamma_{ji})$ is the associated point group element representing a rotation predefined for that edge.

We show that the interaction graph $\G_I$, chosen as a 
spanning tree subgraph of $C_n$, suffices to solve the formation 
control problem. This implies that only $(n-1)$ 
edges are required to guarantee convergence to a 
$\mathcal{C}_n$-symmetric formation.

\subsection{Symmetry-based Control Law}
Similar to the idea presented in \cite{Zelazo2025forced}, we define a \textit{symmetry-forcing potential} over the edges in the interaction graph,
\textcolor{black}{\begin{align}\label{sympotential1}
F(p(t)) = \frac{1}{2} \sum_{ij\in \mathcal{E}_I} \|p_i(t)-\tau(\gamma_{ji} )p_j(t)\|^2.
\end{align}}

To solve the formation control problem \eqref{sym_problem}, we now propose the law defined by the gradient dynamical system
\begin{align}\label{symform_acquire}u(t) = -\nabla F(p(t)).
\end{align}
Then, we obtain the expression of the closed-loop dynamics for each agent $i$:
\begin{align}\label{ea_dyn}
    \dot{p}_i(t)&= \sum_{{ij\in\mathcal{E}_I}}(\tau({\gamma_{ji}})p_j(t)-p_i(t)). 
\end{align}

The closed-loop dynamics of each agent \eqref{ea_dyn} has a straightforward geometric interpretation. The control law attempts to reduce the error vector $p_i-\tau(\gamma_{ji})p_j$, i.e., it attempts to align the state of agent $i$ with the rotated state of its neighbor $j$ according to the prescribed symmetry action. %the distance error between $p_i$ and the term $\tau(\gamma_{ji})p_j$.

We now focus on the closed-loop dynamics in state-space, which has the form
\begin{align}\label{ctrl_1}\dot{p}(t) = - Qp(t),
\end{align}
where $Q\in\R^{2n\times 2n}$ is the resulting \textit{symmetry-constraining} matrix-weighted Laplacian for graph $\G_I$, with the block entries
\textcolor{black}{$$ [Q]_{ij} = \begin{cases}
                d(i)I_2, & i=j, \, i \in \V \\
                -\tau(\gamma_{ji}), & ij\in \mathcal{E}_I \\
                0, & \text{o.w.}
            \end{cases},$$
where $d(u)$ denotes the degree of node $u$ in the induced subgraph $\G_I$}. As a matrix-weighted Laplacian, observe that $Q$ can be expressed as the matrix product $E(\Gamma_r)E(\Gamma_r)^T$, where $E(\Gamma_r)\in\mathbb{R}^{2n \times 2|\E_I|}$ has a matrix-weighted incidence matrix structure with its block-columns being associated with the edge $ij$,
$$[\,\cdots\ \underbrace{I_2}_{\text{node } i}\ \cdots\ \underbrace{-\tau(\gamma_{ji})^T}_{\text{node }j}\ \cdots]^T.$$
\begin{example}
Consider the $\mathcal{C}_4$-symmetric framework seen in Fig. \ref{fig:sym_fw_Cn}(a), with the choice of the edge set $\E_I = \{12, 23,34\}$. Then, the group actions in $\tau(\Gamma_r)$ are rotations about the origin by $2\pi/n=\pi/2$, and the corresponding matrix $Q \in \mathbb{R}^{8\times 8}$ can be expressed as
$$
Q = 
\left[\begin{smallmatrix}
I_2 & -R(\frac{\pi}{2})^T & 0 & 0 \\
-R(\frac{\pi}{2}) & 2I_2 & -R(\frac{\pi}{2})^T & 0 \\
0 & -R(\frac{\pi}{2}) & 2I_2 & -R(\frac{\pi}{2})^T \\
0 & 0 & -R(\frac{\pi}{2}) & I_2
\end{smallmatrix}\right].
$$
\end{example}
Note that $\Null(Q)$ coincides with the set of $\mathcal{C}_4$-symmetric configurations satisfying \eqref{eq:symfwk}. That is, $\tau(\gamma)p_i = p_{\gamma(i)}$ for all $\gamma\in\Gamma_r$ and $i\in \V$.
\begin{proposition}\label{prop:Q_def}
Let $Q$ be the symmetry-constraining matrix-weighted Laplacian associated with $\G_I$. Then, (i) $Q$ is positive semi-definite (PSD). (ii) $Q$ has a nontrivial null-space, $\Null(Q)=\{p\in\mathbb{R}^{2n}| E(\Gamma_r)^T p=0\}$, corresponding to the set of $\mathcal{C}_n$-symmetric configurations, and (iii) the rank of $Q$ is $2n-2$, and $\dim \Null(Q)=2$.
\end{proposition}

\begin{proof}
($i$) Since $Q=E(\Gamma_r)E(\Gamma_r)^T$, for any $p\in\R^{2n}$ we have
$p^TQp=p^TE(\Gamma_r)E(\Gamma_r)^Tp=\|E(\Gamma_r)^Tp\|^2\geq 0.$ Hence, $Q$ is PSD. ($ii$) Without loss of generality, assume the nodes are labeled such that each edge in $\G_I$ is of the form $i(i+1)$.
By construction, $E(\Gamma_r)^Tp=\begin{bmatrix}
r_1^T&\cdots&r_{|\E_I|}^T
\end{bmatrix}^T
\in\mathbb{R}^{2|\E_I|}$ stacks the edge errors $r_i = p_i - \tau(\gamma_{(i(i+1))})^T\,p_{(i+1)}\in\mathbb{R}^2$. Then, $E(\Gamma_r)^T p=0$ iff every edge satisfies $p_i-\tau(\gamma_{ji})p_j=0$. Therefore, by Definition \ref{def:tau-gamma}, $\Null(Q)$ coincides with the respective set of $\mathcal{C}_n$-symmetric configurations. ($iii$) Define $S_1=I_2$ and, for each node $j$, let $S_j\in SO(2)$ denote the ordered product of edge rotations along the unique path from node $1$ to node $j$, so that along each edge $i(i+1)$ we have $S_{i+1}=\tau(\gamma_{i(i+1)})S_i\in SO(2)$.
Let $p_1=q\in\R^2$. Then, if $E(\Gamma_r)^Tp=0$, by substitution we have $p_i \;=\; S_i\,q$ for all $i\in\{1,\dots,n\}$, such that $E(\Gamma_r)^T\begin{bmatrix}(S_1q)^T & (S_2q)^T& \cdots & (S_nq)^T \end{bmatrix}^T = 0$. Note that $q\in\mathbb{R}^2$ has two degrees of freedom. Hence, any vector in $\Null(E(\Gamma_r)^T)=\Null(Q)$ lies in $\IM(V_0)$, %is determined by
%$\IM(V_0)$
where 
\begin{align}\label{v0_def}
    V_0 =\left[
\begin{smallmatrix}
S_1e_1 & S_1e_2\\
\vdots & \vdots\\
S_ne_1 & S_ne_2
\end{smallmatrix}\right],
\end{align}
with $e_1=\begin{bmatrix}1&0\end{bmatrix}^T$ and $e_2=\begin{bmatrix}0&1\end{bmatrix}^T$.
Therefore, $\dim\Null(Q)= 2$, and $\mathrm{rank}(Q) = 2n - \dim\Null(Q) = 2n-2$.
\end{proof}

We now examine the dynamics of the closed-loop system \eqref{ctrl_1} to show that the proposed distributed control law drives the agents from any initial condition to the desired symmetric configuration. We derive the explicit solution of \eqref{ctrl_1}, showing that the limit configuration corresponds to the orthogonal projection of the initial state onto the subspace of $\mathcal{C}_n$-symmetric formations.

\begin{theorem}\label{th1}
    Consider a MAS consisting of $n$ integrator agents \eqref{int-dynamics}, whose interaction topology is defined by $\G_I$, and let
    $$\mathcal{F}=\{p\in\R^{2n}|\tau(\gamma)p_i=p_{\gamma(i)},\;\;\forall\gamma\in\Gamma_r,\;i\in \V\}.$$
    Then, for any initial condition $p(0) \in \mathbb{R}^{2n}$, % $$p(0)=[p_1(0)^T,p_2(0)^T,\ldots,p_n(0)^T]^T$$
    the control \eqref{ctrl_1} renders the set $\mathcal{F}$ exponentially stable, with $p(\infty)$ as the orthogonal projection of $p(0)$ onto $\mathcal{F}$,
\begin{align}\label{ts1}
    \lim_{t\to\infty}p(t)=\frac{1}{n}V_0V_0^{\top}p(0),
\end{align}
where $V_0$ is given in \eqref{v0_def}. Furthermore, the steady-state of each agent is given by 
\begin{align}\label{ts1_ea}
\lim_{t\rightarrow \infty}p_i(t)=\frac{1}{n}S_i\sum_{k=1}^n S_k^Tp_k(0).
\end{align}
\end{theorem}
\begin{proof}
By Proposition \ref{prop:Q_def}, note that $V_0$ is the vertical stack of the $2\times2$ blocks $S_i\begin{bmatrix}
    e_1 & e_2
\end{bmatrix}=S_i$. 
Since $S_i\in SO(2)$ are orthogonal matrices and $V_0^T V_0= \Sigma_{i=1}^n S_i^T S_i = nI_2$,
then the columns of $V_0$ are orthogonal as well. Since $Q$ is PSD and the columns of $V_0$ are orthogonal, we can define $\hat{V}$ as an orthonormal eigenbasis $\hat{V}=\begin{bmatrix}
    \hat{V}_0 & \hat{V}_+
\end{bmatrix}$ with $\hat{V}_0=\frac{1}
{\sqrt{n}}V_0$ and $V_+$ the orthogonal complement of $V_0$.
Hence, $Q$ is equivalent to $Q=\hat{V}\left[\begin{smallmatrix}0 & 0\\0&\Lambda_+\end{smallmatrix}\right]\hat{V}^T,
$ and the closed-form solution of \eqref{ctrl_1} results in $p(t)=e^{-Qt}p(0)=\hat{V}\left[\begin{smallmatrix}I_2 & 0\\0&e^{-\Lambda_+t}\end{smallmatrix}\right]\hat{V}^Tp(0)$. Since all non-zero eigenvalues of $-Q$ are in OLHP, the dynamics of $p(t)$ exponentially converge to $p(\infty)=\hat{V}_0\hat{V}_0^Tp(0)=\frac{1}{n}V_0V_0^{\top}p(0)$. Note that $V_0^Tp(0)=\sum_{k=1}^nS_k^Tp_k(0)$. Then, the block expression for each agent is $p_i(\infty)=\frac{1}{n}S_i\sum_{k=1}^n S_k^Tp_k(0).$
Moreover, observe that $E(\Gamma_r)^Tp=0 $ implies $ p_i=S_iq$ for some $q\in\R^2$. Hence, $\Null(Q)=\IM(V_0)=\mathcal{F}$, rendering $\mathcal{F}$ exponentially stable as claimed.
\end{proof}

\begin{example}\label{ex:c6_ex1}
Consider a MAS consisting of $n=6$ agents, tasked with attaining a $\C_6$-symmetric configuration (see Fig. \ref{fig:sym_fw_Cn}(c)). Fig. \ref{fig:c6_combined}(a) illustrates the underlying $\Gamma$-symmetric graph, with the dashed edge being removed for the chosen communication topology graph $\G_I$ with $5$ edges. Note that by using a distance \cite{Krick2009} or bearing \cite{ZHAO2019_CSM} approach we would require $9$ edges in total to ensure the correct formation shape in $\R^2$.
\begin{figure}[H]
    \vspace{-0.25cm}
    \centering
    \begin{subfigure}{0.35\linewidth}
        \centering        \includegraphics[width=0.6\linewidth]{figures/c6_sim_gr.tex}
        \vspace{0.4cm}
        \hspace{-0.3cm}
        \caption{Underlying cycle graph $C_6$.}
        \label{fig:sim1_tg} 
        \vspace{-0.1cm}
    \end{subfigure}
    \hfill
    \begin{subfigure}{0.6\linewidth}
        \centering
        \includegraphics[width=0.8\linewidth]{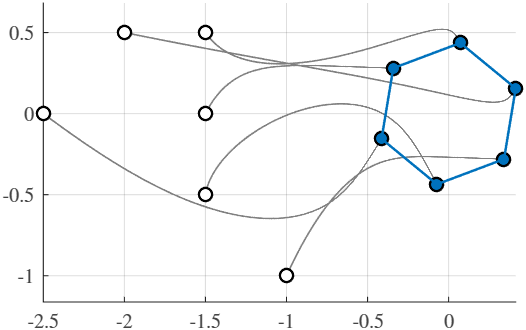}
        \caption{Trajectories generated from \eqref{ctrl_1}.}
        \label{fig:c6_st1_fig}
        \vspace{0.25cm}
    \end{subfigure}

    \caption{Graph and simulation results for Example~\ref{ex:c6_ex1}.}
    \label{fig:c6_combined}
\end{figure}
    \vspace{-0.25cm}
Fig. \ref{fig:c6_combined}(b) illustrates the trajectories and resulting configuration obtained by implementing the proposed control law \eqref{ctrl_1}, where the corresponding matrix $Q$ is constructed using the rotation matrices $\tau(\gamma_{ij})=R(\pi/3)$. 
\end{example}

\section{Formation Maneuvering}\label{sec.maneuver}

Our main focus has been on achieving and maintaining a target formation shape. Observe in Fig. \ref{fig:c6_st1_fig} that, due to symmetry, the control \eqref{ctrl_1} successfully drives the agents to a desired formation shape, but with respect to an inertial (global) origin. This may be limiting in many practical scenarios where the formation requires the ability to maneuver, that is, translate, rotate, and scale while preserving the desired shape. To improve the flexibility of a symmetry-based formation control approach, 
\cite{Zelazo2025forced} proposed augmenting the closed-loop dynamics for each agent \eqref{ea_dyn} with a virtual state $r(t)$ that enables the agents to agree on a different origin. We leverage this idea to address formation maneuvering as well.

\begin{assumption}\label{as_mv}
    Each agent in the MAS has access to a virtual trajectory predefined by a time-varying
    \begin{itemize}
        \item[i)] translation $r(t)\in\R^2$ with $\dot{r}(t)=v(t)$;
        \item[ii)] rotation $\mathcal{R}(t)\in SO(2)$ with $\dot{\mathcal{R}}(t)=\Omega(t)\mathcal{R}(t)$ where $\Omega(t) = \left[\begin{smallmatrix}
            0 & -\omega(t) \\ \omega(t) & 0
        \end{smallmatrix}\right]$, and $\omega(t)$ is the desired angular velocity of the formation;
        \item[iii)] scale factor $s(t)\in\R^+$, with $\dot{s}(t)=\alpha(t)s(t)$, $\alpha(t) \in\R$.
    \end{itemize}
\end{assumption}

Reference trajectories are known a priori in many applications \cite{Queiroz2019}. Therefore, building on Assumption \ref{as_mv}, we define a shifted state for each agent, $c_i(t)=p_i(t)-r(t)$. Moreover, since the formation is specified with respect to a fixed inertial point, we define the centroid of the formation at the origin with the axis of rotation passing through it. Under Assumption \ref{as_mv}, we then propose the augmented control law
\begin{align}\label{ctrl_2}
    \hspace{-7pt}u(t)=&-Qc(t)+\mathds{1}_n\!\otimes \! v(t)+ 
    (I_n \! \otimes \! \Omega(t)+\alpha(t)) c(t),
\end{align}
which enables the agents to converge to the desired configuration while maneuvering along the predefined trajectory.
\begin{theorem}
 Consider a MAS consisting of $n$ integrator agents \eqref{int-dynamics} satisfying Assumption \ref{as_mv}, whose interaction topology is defined by $\G_I$, and let
    $$\mathcal{F}_c=\{p\in\R^{2n}|\tau(\gamma)c_i=c_{\gamma(i)},\;\;\forall\gamma\in\Gamma_r,\;i\in \V\}$$
    be the set of all shifted $\mathcal{C}_n$-symmetric configurations.
    Then, for any initial condition $p(0) \in \mathbb{R}^{2n}$, % $$p(0)=[p_1(0)^T,p_2(0)^T,\ldots,p_n(0)^T]^T$$
    the control \eqref{ctrl_2} renders the set $\mathcal{F}_c$ exponentially stable. 
\end{theorem}
\begin{proof}
Define $\zeta(t)\in\R^{2n}$ as the configuration $p(t)\in\R^{2n}$ expressed in a frame moving along the virtual trajectory,
\begin{align}\label{moving_dynamics}
\zeta(t)=\frac{1}{s(t)}\big(I_n \otimes \mathcal{R}(t)^T\big)c(t)\in\mathbb R^{2n}.
\end{align}

We examine the derivative of each agent state $\zeta_i(t)\in\R^2$ (product rule),
{\small\begin{align*}
\dot{\zeta}_i(t)&=-\frac{\dot{s}(t)}{s^2(t)}\mathcal{R}(t)^Tc_i(t)+\frac{1}{s(t)}\big(\dot{\mathcal{R}}(t)^Tc_i(t)
     +\mathcal{R}(t)^T\dot{c_i}(t)\big).
\end{align*}}
Note that 
$\dot{\mathcal{R}}(t)^T=-\mathcal{R}(t)^T\Omega(t)$. Since $\Omega(t)$ and $\mathcal{R}(t)$ commute in $\R^2$, we have $\dot{\mathcal{R}}(t)^T=-\Omega(t)\mathcal{R}(t)^T$. Then 
{\small\begin{align*}
    \dot{\zeta}_i(t)&=-\alpha(t)\zeta_i(t)-\Omega(t)\zeta_i(t)+\frac{1}{s(t)}\big(\mathcal{R}(t)^T(\dot{u}_i(t)-v(t))\big).
\end{align*}}
By applying the control law \eqref{ctrl_2}, we have
{\small \begin{align*}
    \dot{\zeta}_i(t)=-\alpha(t)\zeta_i(t)-\Omega(t)\zeta_i(t)-\frac{1}{s(t)}\mathcal{R}(t)^Tv(t)+\frac{1}{s(t)}\mathcal{R}^T\\\big(\sum_{{ij\in\E_I}}(\tau({\gamma_{ji}})c_j(t)-c_i(t))+v(t)+\Omega(t) c_i(t) + \alpha(t)c_i(t)\big).
\end{align*}}
Since $\zeta_i(t)=\frac{1}{s(t)}\mathcal{R}(t)^Tc_i(t)$, all trajectory dependent terms cancel, simplifying the expression to 
\begin{align*}
    \dot{\zeta}_i(t)=\sum_{{ij\in\E_I}}(\tau({\gamma_{ji}})\zeta_j(t)-\zeta_i(t)),
\end{align*}
reducing the analysis of the dynamics to $\dot{\zeta}(t)=-Q\zeta(t)$. By Theorem \ref{th1}, the dynamics of $\zeta(t)$ ensure that the formation exponentially converges to the set 
$$\mathcal{F}_\zeta=\{\zeta\in\R^{2n}|\tau(\gamma)\zeta_i=\zeta_{\gamma(i)},\;\;\forall\gamma\in\Gamma_r,\;i\in \V\}.$$
From the definition of $\zeta_i(t)$ \eqref{moving_dynamics}, this set is equivalent to $\mathcal{F}_c$, rendering the set $\mathcal{F}_c$ exponentially stable as claimed.
\end{proof}

\begin{example}\label{ex:c6_ex2}
Consider the same setup as in Example \ref{ex:c6_ex1} under Assumption \ref{as_mv}. A trajectory is predefined to enable the formation to maneuver through obstacles along a desired path. The blue line in Fig. \ref{fig:c6_traj_fig} illustrates the translational trajectory along the path, and the scaled arrows indicate the rotation and scaling states with respect to the initial state.

\begin{figure}[ht]
\hspace{-0.1cm}
    \includegraphics[width=0.49\textwidth]{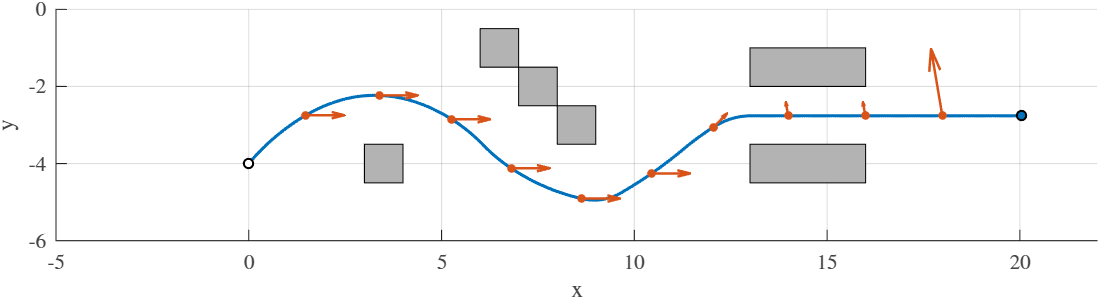}
    \caption{Predefined reference trajectory used in control law \eqref{ctrl_2}.}
    \label{fig:c6_traj_fig}
\end{figure}

Fig. \ref{fig:c6_mv1_fig} illustrates the resulting agent trajectories along the predefined virtual trajectory under control law \eqref{ctrl_2}.

\begin{figure}[ht]
\hspace{-0.15cm}
    \includegraphics[width=0.48\textwidth]{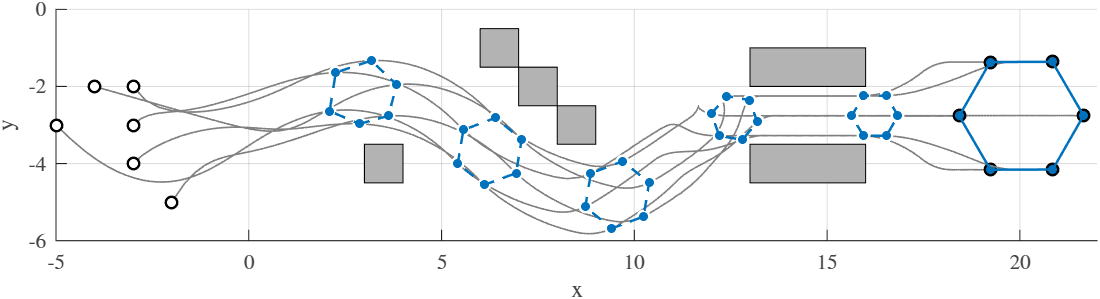}
    \caption{Agent trajectories generated by \eqref{ctrl_2} along the predefined reference trajectory.}
    \label{fig:c6_mv1_fig}
\end{figure}

To further evaluate the system, Fig. \ref{fig:c6_err_sys} illustrates the inter-agent rotation symmetry errors during maneuvering \eqref{sym_problem}. The errors exponentially converge to zero, showing the effectiveness of the proposed method.

\begin{figure}[ht]
\hspace{-0.145cm}
    \includegraphics[width=0.49\textwidth]{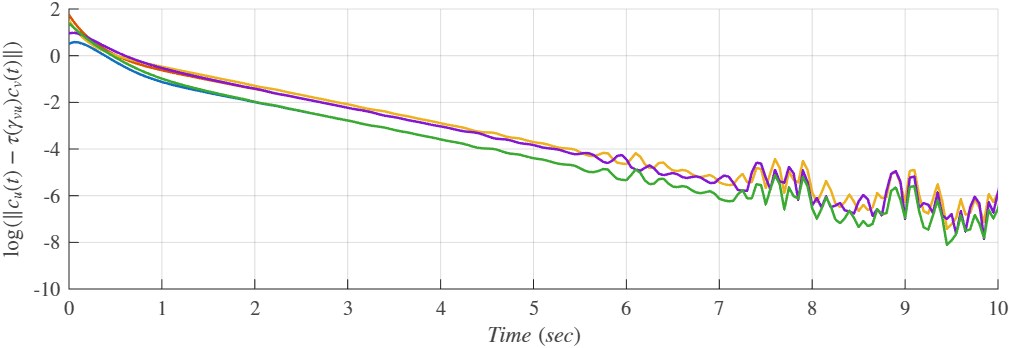}
    \caption{Norm of the symmetry error over time.}
    \label{fig:c6_err_sys}
\end{figure}

\end{example}
\section{Extension to $\R^3$}\label{sec.R3}
The results of Section \ref{sec.maneuver} are stated explicitly for formations in $\R^2$. Extending these ideas to $\R^3$ is relatively straightforward. In this section, we provide a numerical example to illustrate how this might be done, with a formal analysis reserved for future work.

Consider a MAS consisting of $n=8$ agents required to achieve a cube formation with an underlying graph $\G$ shown in Fig. \ref{fig:sim_cube_gr}(a). The target configuration can again be defined as a $\tau(\Gamma)$-symmetric framework, where the elements of $\tau(\Gamma)$ are given by rotation matrices $R\in SO(3)$. In $\R^3$, such rotations are defined about coordinate axes, or general hyperplanes through the origin. For a cube formation, a natural choice of symmetries is given by $\mathcal{C}_4$ rotations about the coordinate axes (see Fig. \ref{fig:sim_cube_gr}(b)). For instance, the agents $\{p_1,\, p_2,\, p_3,\, p_4\}$ and $\{p_5,\, p_6,\, p_7,\, p_8\}$ may each satisfy a $\mathcal{C}_4$-symmetric framework about the $z$-axis, and agents $\{p_2,\, p_1,\, p_5,\, p_6\}$ can be constrained to form a $\mathcal{C}_4$-symmetric framework orthogonal to the $z$-axis. %\todo[inline]{would it help to include $x-y-z$ axis in the figure? \checkmark} 
These symmetry relations suffice for agents to exchange information according to the subgraph $\G_I$, obtained by removing the dashed edges as shown in Fig. \ref{fig:sim_cube_gr}. By construction, $\G_I$ is also a spanning tree. Similar to the planar case, we now define the symmetry-constraining matrix $Q_z$ defined by $\mathcal{C}_4$ symmetries about the $z$-axis, and $Q_\perp$ as the corresponding matrix defined by $\mathcal{C}_4$ symmetries about an axis orthogonal to the $z$-axis. The resulting symmetry-constraining matrix $Q$ for the cube formation is then obtained as a composition
$$
Q=I_2\otimes Q_z + P\begin{bmatrix}
    Q_\perp^T & 0
\end{bmatrix}^TP^T,
$$
where $P$ is a permutation matrix that reorders the block structure of $Q_\perp$ so that the composition of $Q_z$ and $Q_\perp$ matches the indexing of the stacked state vector in the control law.

\begin{figure}[ht]
\begin{center}
\vspace{-0.5cm}
\includegraphics[width=0.55\linewidth]{figures/cube_sub_gr.tex}
    \end{center}
\vspace{-0.2cm}
\caption{(a) Shows the underlying graph $\G$, and (b) shows the desired $\tau(\Gamma)$-symmetric framework of the cube formation.}
\label{fig:sim_cube_gr}
\end{figure}
\vspace{-0.2cm}
By construction, $Q\succeq0$, and its null-space corresponds to the set of cube-symmetric configurations. Hence the control \eqref{ctrl_1} drives the system exponentially to the desired configuration. Similar to the planar case, this method can be augmented with a virtual reference trajectory $(r(t),\rr(t),s(t))$. However, note that in $\R^3$, $\Omega(t)$ and $\rr(t)$ do not always commute. Hence, during rotations, we also require the matrix $Q$ to undergo a similarity transformation according to $\rr(t)$. Fig. \ref{fig:cube_fig} shows the resulting trajectories obtained by implementing control law \eqref{ctrl_2}. The agents converge to a cube formation along the predefined trajectory.
\begin{figure}[ht]
\vspace{-0.3cm}
\hspace{-0.1cm}\centering 
\includegraphics[width=0.48\textwidth]{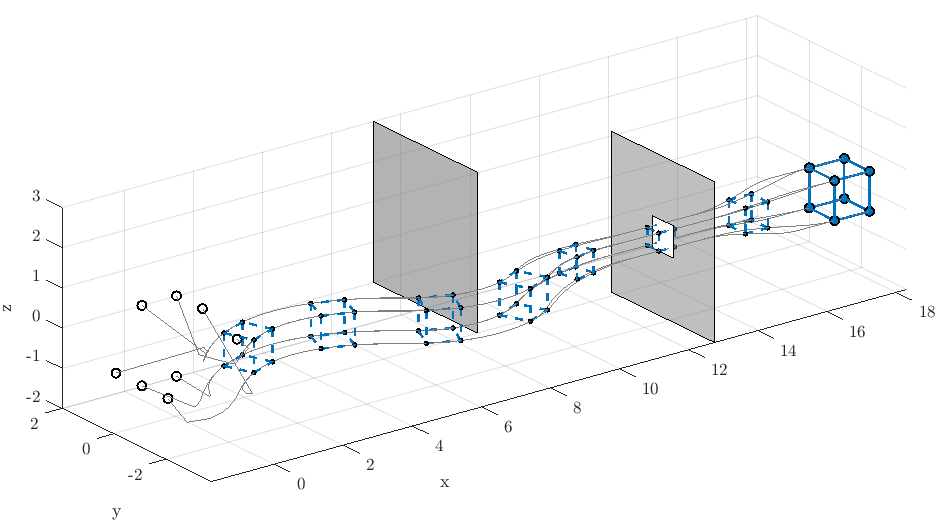}
    \caption{Cube-formation trajectories generated by \eqref{ctrl_2} along a predefined reference trajectory.}
    \label{fig:cube_fig}
\end{figure}
\section{Concluding Remarks}\label{sec.conclusion}
This paper introduced a formation control strategy that achieves cyclic target configurations using only rotation symmetry constraints between neighboring agents. In particular, we showed that a communication spanning-tree subgraph with $n-1$ constraints, matching the minimal connectivity requirement, is sufficient to implement the protocol. By augmenting the control law with a time-varying virtual state, we further demonstrated that the formation can perform coordinated translations, rotations, and scalings. We also presented a numerical extension of the method to formations in $\R^3$. Overall, the results highlight the potential of symmetry-based constraints for formation control. Future work will focus on a formal extension to broader point-group elements in $\R^3$, directed and switching interaction graphs, and leader-follower architectures, with the goal of enabling fully distributed agreement on time-varying virtual trajectories

\bibliographystyle{IEEEtran}
\bibliography{references}

%%%%%%%%%%%%%%%%%%%%%%%%%%%%%%%%%%%%%%%%%%%%%%%%%%%%%%%%%%%%%%%%%%%%%%%%%%%%%%%%
\end{document}